\documentclass[11pt,letterpaper]{article}

\usepackage{typearea}
\typearea{15}

\usepackage{theorem,latexsym,graphicx}
\usepackage{amsmath,amssymb,enumerate}
\usepackage{xspace}
\usepackage{bm}
\usepackage{ifpdf}
\usepackage{shadow,shadethm,color}
\usepackage[compact]{titlesec}
\usepackage{algorithm}
\usepackage[noend]{algorithmic}
\usepackage{subfigure}
\usepackage{verbatim}
\usepackage{paralist}

\allowdisplaybreaks

\definecolor{Darkblue}{rgb}{0,0,0.4}
\definecolor{Brown}{cmyk}{0,0.81,1.,0.60}
\definecolor{Purple}{cmyk}{0.45,0.86,0,0}
\newcommand{\mydriver}{hypertex}
\ifpdf
 \renewcommand{\mydriver}{pdftex}
\fi
\usepackage[breaklinks,\mydriver]{hyperref}
\hypersetup{colorlinks=true,%pdfborder={1 1 1 [3]},%
            citebordercolor={.6 .6 .6},linkbordercolor={.6 .6 .6},%
citecolor=Darkblue,urlcolor=black,linkcolor=Darkblue,pagecolor=black}
\newcommand{\lref}[2][]{\hyperref[#2]{#1~\ref*{#2}}}


\newtheorem{theorem}{Theorem}[section]

\newtheorem{lemma}[theorem]{Lemma}
\newshadetheorem{lemmashaded}[theorem]{Lemma}

\numberwithin{algorithm}{section}

\newenvironment{proof}{

\noindent{\bf Proof:}}
{\hfill$\blacksquare$

}

\newcommand{\junk}[1]{}
\newcommand{\ignore}[1]{}

\newcommand{\onGAP}{{\textsf{OnGAP}}\xspace}

\newcommand{\last}{{\psi}}

\newcounter{note}[section]

\newcommand{\qedsymb}{\hfill{\rule{2mm}{2mm}}}
\renewenvironment{proof}{\begin{trivlist} \item[\hspace{\labelsep}{\bf
\noindent Proof.\/}] }{\qedsymb\end{trivlist}}%

\newcommand{\initOneLiners}{%
    \setlength{\itemsep}{0pt}
    \setlength{\parsep }{0pt}
    \setlength{\topsep }{0pt}
}

\newcommand{\squishlist}{
 \begin{list}{$\bullet$}
  { \setlength{\itemsep}{0pt}
     \setlength{\parsep}{3pt}
     \setlength{\topsep}{3pt}
     \setlength{\partopsep}{0pt}
     \setlength{\leftmargin}{1.5em}
     \setlength{\labelwidth}{1em}
     \setlength{\labelsep}{0.5em} } }

\newcommand{\squishend}{
  \end{list}  }

\newcommand{\E}{\mathbb{E}}

\begin{document}
\title{Online Primal-Dual For Non-linear Optimization \\with Applications
  to Speed Scaling}

\author{
Anupam Gupta\thanks{Computer Science Department, Carnegie Mellon
    University, Pittsburgh, PA 15213, USA. Supported in part by
    NSF awards CCF-0964474 and CCF-1016799. R.K.\ supported in part by
    an IBM Graduate Fellowship.}
\and Ravishankar Krishnaswamy\footnotemark[1]
\and Kirk Pruhs\thanks{
Computer Science Department,
University of Pittsburgh, Pittsburgh, PA 15260, USA.
{\tt kirk@cs.pitt.edu}.
Supported in part by
NSF grant CCF-0830558 and
an IBM Faculty Award.}
}
\date{}
\maketitle
\thispagestyle{empty}
%\vspace{-7mm}

\begin{abstract}
  We reinterpret some online greedy algorithms for a class of nonlinear
  ``load-balancing'' problems as solving a mathematical program
  online. For example, we consider the problem of assigning jobs to
  (unrelated) machines to minimize the sum of the $\alpha^{th}$-powers
  of the loads plus assignment costs (the \emph{online Generalized
    Assignment Problem}); or choosing paths to connect terminal pairs to
  minimize the $\alpha^{th}$-powers of the edge loads (i.e.,
  \emph{online routing with speed-scalable routers}).  We give analyses
  of these online algorithms using the dual of the primal program as a
  lower bound for the optimal algorithm, much in the spirit of online
  primal-dual results for linear problems.

  We then observe that a wide class of uni-processor speed scaling
  problems (with essentially arbitrary scheduling objectives) can be
  viewed as such load balancing problems with linear assignment
  costs. This connection gives new algorithms for problems that had
  resisted solutions using the dominant potential function approaches
  used in the speed scaling literature, as well as alternate, cleaner
  proofs for other known results.
\end{abstract}

\vfill
\thispagestyle{empty}
\setcounter{page}{0}
\pagebreak

\section{Introduction}
\label{sec:introduction}

In this paper, we consider two online problems related to load balancing.
We call the first problem \emph{Online
  Generalized Assignment Problem (\onGAP)}:

{\bf Definition of \onGAP:}  Jobs arrive one by one in an online manner, and the algorithm must fractionally assign these jobs to one of $m$
machines. When a job $j$ arrives, the online algorithm learns
$\ell_{je}$, the amount by which the load of machine $e$ would increase for \emph{each unit} of work of job $j$ that is assigned to machine $e$,
and $c_{je}$, the assignment cost incurred for each unit of
work of job $j$ that is assigned to machine $e$. The goal is to minimize
the sum of the $\alpha^{th}$ powers of the machine loads, plus the total
assignment cost. 

The version of \onGAP without assignment costs was studied by
\cite{AAFPW,AAGKKV95}. Our original motivation for studying \onGAP is
that it models a well-studied class of \emph{speed scaling problems with
  sum cost scheduling objectives}. In these problems, jobs arrive over
time and must be scheduled on a speed scalable processor---i.e., a
processor that can run at any non-negative speed, and uses power
$s^\alpha$ when run at speed $s$. The objective is the sum of the energy
used by the processor plus a fractional scheduling objective that is the
sum over jobs of the ``scheduling cost'' of the individual jobs.  These
speed scaling problems are a special case of \onGAP where the machines
model the times that jobs can be scheduled, the assignment cost $c_{je}$
models the scheduling cost for scheduling a unit of job $j$ at time $e$.
For example, one such scheduling objective is the sum of the fractional
flow/response times squared. For this objective, $c_{je}$ is $(e -
r_j)^2$ for all times $e$ that are at least the release time $r_j$ of
job $j$, and infinite otherwise. Another example is the problem of
minimizing energy usage subject to deadline constraints, introduced by
\cite{YDS} and considered in followup papers \cite{BKP, BansalBCP11,
  BansalCPK09}. This problem can be viewed as a special case of \onGAP,
where each job $j$ has an associated deadline $d_j$, and $c_{je}$ is
zero if $e \in [r_j, d_j]$ and infinite otherwise.

The second problem that we consider is a variation/generalization of \onGAP involving
\emph{online routing with speed scalable routers} to minimize energy,
which was previously considered in \cite{AAFPW, AndrewsAZ11}. 

{\bf Definition of Online Routing with Speed Scalable Routers Problem:} 
A sequence of requests arrive one by one over time.
Each request $j$ has an associated source-sink pair $(s_j, t_j)$
in a network of speed scalable routers, and the online algorithm must
route flow between the source-sink pair, with an objective of minimizing
the total energy used by the network, where the energy incurred by an edge $e$ is the $\alpha^{th}$ power of the load flowing through it.

For load balancing and online routing, it was known that natural online
greedy algorithms, which assign jobs to the machine(s) that minimize the
increase in cost, can be shown to be $O_\alpha(1)$-competitive via an
exchange argument, and directly bounding the cost compared to the
optimal cost~\cite{AAFPW,AAGKKV95}.  (In fact, basically the same
argument shows that the online greedy algorithm is
$O_\alpha(1)$-competitive for \emph{integer} assignments.)  Once we
observe that speed scaling problems with sum scheduling objectives can
be reduced to \onGAP, it is not too difficult to see that the analysis
technique in~\cite{AAGKKV95} can be used to show that natural greedy
speed scaling algorithms are $O_p(1)$-competitive.

% \kirknote{Why did we mention that the problems are convex? What do we care? We can handle nonconvex
% problems as well.}

{\bf Our Contribution.} In this paper, we first interpret these online
problems as solving a mathematical program online, where the constraints
arrive one-by-one, and in response to the arrival of a new constraint,
the online algorithm has to raise some of the primal variables so that
the new constraint will be satisfied. The online algorithms that we
consider raise the primal variables greedily. Our competitive analysis
will use the dual function of the primal program as a lower bound for
optimal. For analysis purposes, we assign a value to the dual variable
corresponding to a constraint after the online algorithm has satisfied
that constraint. Our goal is to set the dual variables so that the
resulting dual solution is closely related to the online solution. (In
our analyses, the settings of the dual variables naturally correspond to
(some approximation for) the increase in the objective function.) We
first show how to obtain fractional solutions to these problems, and
subsequently show how similar ideas  can be used for integer
assignments.
%As is often the case in such
%primal-dual proofs, the key step is figuring out how to set the dual
%variables.
%Kirk says, I don't mind if we put the above sentences, back, but I would say the
% key is both setting the duals, and finding the greedy algorithm that solves the dual problem.

Our analyses are very much in the spirit of the online primal dual
technique for linear programs~\cite{BN-mono}.  The main difference is
that in the nonlinear setting, the dual is more complicated than in the
linear setting (where the dual is just another linear program). Indeed,
in the nonlinear setting, one can not disentangle the objective and the
constraints, since the dual itself contains a version of the primal
objective, and hence copies of the primal variables, within
it. Consequently, the arguments for the dual function in the nonlinear
setting have a somewhat different feel to them than in the linear
setting. In particular, we need to set the dual variables $\lambda$, and
then find minimizing settings for the copies of the primal variables to
get a good lower bound. For the load-balancing and speed scaling
problems, this proceeds relatively naturally. But for the routing
problem the dual minimization problem is itself non-trivial: in this
case we first show how to write a ``relaxed/decoupled'' dual, which is
potentially weaker than the original dual, but easier to argue about,
and then set the variables of this relaxed dual to achieve a good lower
bound. We hope this analytical technique will be useful for other
problems. 

% While the bounds on the competitive ratio that we achieve using duality
% are not better than what one obtains using the methods in
% \cite{AAFPW,AAGKKV95}. However, 

We then show how a wide class of uni-processor speed-scaling problems
(with essentially arbitrary scheduling objectives) can be viewed as load
balancing problems with linear assignment costs. This connection gives
new algorithms for speed-scaling problems that had resisted solutions
using the dominant potential function approaches used in the speed
scaling literature, as well as alternate, cleaner analyses for some
known results. For speed scaling problems, our analysis using duality is
often cleaner (compare for example, the analysis of OA in \cite{BKP} to
the analysis given here) and more widely applicable (for example, to
nonlinear scheduling objectives) than the potential function-based
analyses.  Furthermore, we believe that much like the online primal-dual
approach for linear problems, the techniques presented here have
potential for wide applicability in the design and analysis of online
algorithms for other non-linear optimization problems.

\paragraph{Roadmap:} In \lref[Section]{sec:related-work} we discuss
related work.  In \lref[Section]{sec:lpnorm} we consider \onGAP. In
\lref[Section]{sec:speed-scaling}, we make some comments about the
application of these results to speed scaling problems.  In
\lref[Section]{sec:routing} we consider the online routing problem.  In
\lref[Section]{sec:rounding} we show how to alter the water-filling
algorithm to obtain \emph{integer assignments} with a similar
competitive ratio, as well a simple randomized rounding with a slightly
worse performance.

\subsection{Related Work}
\label{sec:related-work}

An $O(\alpha)$-competitive online greedy algorithm for the unrelated
machines load-balancing problem in the $L_\alpha$-norm was given by
\cite{AAFPW,AAGKKV95}; Caragiannis~\cite{Carag08} gave better analyses
and improvements using randomization. An offline $O(1)$-approximation
for this problem was given by~\cite{AE05} and~\cite{KMPS05}, via
solving the convex program and then rounding the solution in a
correlated fashion.  For the routing problem, the
$O(\alpha)^\alpha$-algorithm can be inferred from the ideas
of~\cite{AAFPW,AAGKKV95}.  Followup work in a setting of a network
consisting of routers with static and dynamic power components can be
found in~\cite{AndrewsAZ10,AndrewsAZ11}.

There are two main lines of speed scaling research that fit within the
framework that we consider here. This first is the problem of minimizing
energy used with deadline feasibility constraints. \cite{YDS} proposed two online algorithms
\emph{OA} and \emph{AVR}, and showed that \emph{AVR} is
$O_\alpha(1)$-competitive by reasoning directly about the optimal
schedule.  \cite{BKP}~introduced the use of potential functions for
analyzing online scheduling problems, and showed that \emph{OA} and another algorithm \emph{BKP} are
$O_\alpha(1)$-competitive. \cite{BansalBCP11} gave a potential function
analysis to show that \emph{AVR} is $O_\alpha(1)$-competitive.
\cite{BansalCPK09} introduced the algorithm \emph{qOA}, and gave a potential function
analysis to show that
it has a better competitive ratio than OA or AVR for smallish $\alpha$.

The second main line of speed scaling research is when the scheduling
objective is total flow, or more generally total weighted flow.
\cite{PUW,AF} gave offline algorithms for unit-weight unit-work jobs.
All of the work on online algorithms consider some variation of the ``natural''
algorithm, which uses the ``right'' job selection algorithm from the literature on scheduling
fixed speed processors, and sets the power of the processor equal to the weight of the
outstanding jobs. This speed scaling policy is ``natural'' in that it balances the energy and scheduling costs.
By reasoning directly about the energy optimal schedule, \cite{AF}
showed that a batched version of the natural algorithm is
$O_\alpha(1)$-competitive for unit-work unit-weight jobs.  Using a
potential function analysis, \cite{BPS} showed that a variation of the
natural algorithm is $O_\alpha(1)$-competitive for arbitrary-weight
arbitrary-work jobs.  For the objective of total flow plus energy,
the bound on the competitive ratio was improved in
\cite{LLTW08} by use of potential function tailored to integer flow instead of
fractional flow. Using a potential function analysis, \cite{BansalCP09}
showed a variation on the natural algorithm is $O(1)$-competitive
for total flow plus energy for an arbitrary power function,
and a variation on the natural algorithm is scalable, for
fractional weighted flow plus energy for an arbitrary power
function. \cite{AndrewLW10} improved the bound on the competitive ratio
for total flow plus energy. Nonclairvoyant algorithms are analyzed in \cite{ChanELLMP09,ChanLL10}.
A relatively recent survey of the algorithmic power management literature  in general, and 
the speed scaling literature in particular, can be found in \cite{Albers}.

An extensive survey/tutorial on the online primal dual technique for linear problems,
as well the history of the development of this technique, can be found in \cite{BN-mono}.

\section{The Online Generalized Assignment
Problem}
\label{sec:lpnorm}

In this section we consider the problem of Online Generalized Assignment
Problem (\onGAP).  If $x_{je}$ denotes the extent to which job $j$ is
assigned on machine~$e$, then this problem can be expressed by the
following mathematical program:
\begin{align*}
\min\quad& \sum_e \bigg(\sum_{j} \ell_{je} x_{je}\bigg)^\alpha  + \sum_e \sum_j c_{je} x_{je} \\
\textrm{subject to}\quad&\sum_{e} x_{je} \geq 1 \qquad j = 1, \ldots, n
\end{align*}
The dual function of the primal relaxation is then
\begin{align}
g(\lambda) = \min_{x \succeq 0}
 \bigg( \sum_j \lambda_j   +    \sum_e \bigg(\sum_{j} \ell_{je} x_{je}\bigg)^\alpha   +  \sum_{j,e} c_{je} x_{je}          -   \sum_{j,e} \lambda_j\, x_{je} \bigg) \label{3eq:g}
\end{align}
One can think of the dual problem as having the same instance as the primal, but where jobs are allowed
to be assigned to extent zero. In the objective, in addition to the same load cost
$\sum_e \big(\sum_{j} \ell_{je} x_{je}\big)^\alpha$ as in the primal, a fixed cost of $\lambda_j$ is paid for each job
$j$, and a payment of $\lambda_j - c_{je}$ is obtained for each unit of job $j$ assigned.
It is well known that each feasible value of the dual function is a lower bound to the optimal primal solution;
this is  \emph{weak duality}~\cite{Boyd}.

\medskip
\noindent
{\bf Online Greedy Algorithm Description:}
Let $\delta$ be a constant that we will later set to $\frac{1}{\alpha^{\alpha -1}}$.
To schedule job $j$, the load is increased on the machines for which the increase the cost will be the least,
assuming that energy costs are discounted by a factor of $\delta$, until a unit of job $j$
is scheduled. 
More formally, the value of all the primal variables $x_{je}$ for all the
 machines $e$ that minimize
  \begin{equation}
\label{derivative} \delta \cdot {\alpha \cdot \ell_{je}}  \bigg(\sum_{i \leq j} \ell_{ie} x_{ie}\bigg)^{\alpha-1} +  c_{je} \end{equation}
are increased until all the work from job $j$ is scheduled, i.e., $\sum_{e} x_{je} = 1$. Notice that 
${\alpha \cdot \ell_{je}}  \big(\sum_{i \leq j} \ell_{ie} x_{ie}\big)^{\alpha-1}$ is the rate at which the load cost 
is increasing for machine $e$, and $ c_{je}$ is the rate that assignment costs are increasing for machine $e$. In other words, our algorithm fractionally assigns the job to the machines on which the overall objective function increases at the least rate. Furthermore,
observe that if the algorithm begins assigning the job to some machine~$e$, it does not stop raising the primal variable $x_{je}$ until the job is fully assigned\footnote{It may however increase $x_{je}$ and $x_{je'}$ at different rates so as to balance the derivatives where $e$ and $e'$ are both machines which minimize equation~\ref{derivative}}. By this monotonicity property, it is clear that all machines $e$ for which $x_{je} > 0$ have the same value of the above derivative when $j$ is fully assigned.
Now, for the purpose of analysis, we set the value $\widehat{\lambda}_j$ to be the rate of increase of the objective value when we assigned the last infinitesimal portion of job $j$.
More formally, if $e$ is any machine on which job $j$ is run, i.e.,  if $x_{je} > 0$, then
\begin{equation}
  \widehat{\lambda}_j := \delta \cdot {\alpha \cdot \ell_{je}} \bigg(\sum_{i\leq j}
  \ell_{ie} x_{ie}\bigg)^{\alpha-1} +  c_{je}
\end{equation}
Intuitively, $\widehat{\lambda}_j$ is a surrogate for the total increase in objective function value due to our fractional assignment of job $j$ (we assign a total of $1$ unit of job $j$, and $\lambda_j$ is set to be the rate at which objective value increases). 

We now move on to the analysis of our algorithm. To this end, let $\widetilde{x}$ denote the final
value of the $x_{je}$ variables for the online algorithm.

\medskip \noindent {\bf Algorithm Analysis.}
To establish that the online algorithm is $\alpha^\alpha$-competitive, note that
it is sufficient (by weak duality) to show that $g(\widehat{\lambda})$ is at least $\frac{1}{\alpha^\alpha}$ times the cost of the online solution.
Towards this end, let $\widehat x$
be the value of the minimizing $x$ variables in $g(\widehat
\lambda)$, namely
$$\widehat{x} = \arg\min_{x \succeq 0}  \bigg( \sum_j \widehat \lambda_j +  \sum_e \bigg(
      \sum_{j} \ell_{je}x_{je}\bigg)^\alpha - \sum_{j,e} \bigg( \widehat \lambda_j - {c_{je}}\bigg)\, x_{je}  \bigg)$$
      Observe that the values $\widehat x$
  could be very different from the values  $\widetilde{x}$, and indeed
  the next few Lemmas try to characterize these values.
  \lref[Lemma]{3lem:one-job} notes that $\widehat x$ only has one job $\varphi(e)$ on each machine~$e$, and \lref[Lemma]{3lem:phi-job} shows how to determine $\varphi(e)$ and $\widehat{x}_{\varphi(e)e}$.  Then, in \lref[Lemma]{3lem:last-phi}, we show that a feasible choice for the job $\varphi(e)$ is the latest arriving job for which the online algorithm scheduled some bit of work on machine~$e$; Let us denote this latest job by $\last(e)$.

\begin{lemma}
  \label{3lem:one-job}
There is a minimizing solution $\widehat{x}$ such that if $\widehat{x}_{je}  > 0$, then $\widehat{x}_{ie} = 0$ for all $i \ne j$.
\end{lemma}
\begin{proof}
  Suppose for some machine $e$, there exist distinct jobs $i$ and $k$
  such that $\widehat{x}_{ie} > 0$ and $\widehat{x}_{ke} > 0$. Then by
  the usual argument of either increasing or decreasing these variables
  along the line that keeps their sum constant, we can keep the convex
  term $(\sum_{j} \ell_{je}\widehat{x}_{je})^\alpha$ term fixed and not
  increase the linear term $\sum_{j} ( \widehat \lambda_j - {c_{je}})\,
  \widehat{x}_{je}$. This allows us to either set $\widehat{x}_{ie}$ or
  $\widehat{x}_{ke}$ to zero without increasing the objective.
\end{proof}

\begin{lemma}
  \label{3lem:phi-job}
  Define $\varphi(e) = \arg\max_{j} \frac{\left(\widehat \lambda_j -
      {c_{je}} \right)}{\ell_{je}}$.  Then $\widehat{x}_{\varphi(e)e} =
  \frac{1}{\ell_{\varphi(e)e}} \left( \frac{\widehat
      \lambda_{\varphi(e)} - { c_{\varphi(e)e}}}{\alpha \ell_{ \varphi(e)
        e}}\right)^{1/(\alpha - 1)}$ and $\widehat{x}_{je} = 0$ for $j \neq
  \varphi(e)$. Moreover, the contribution of machine $e$ towards $g(\widehat{\lambda})$ is exactly $(1-\alpha) \left(
    \frac{\widehat \lambda_{\varphi(e)} - {c_{\varphi(e)e}}}{\alpha \ell_{\varphi(e)e}} \right)^{\alpha/(\alpha -
    1)} \,$.
\end{lemma}

\begin{proof}
  By \lref[Lemma]{3lem:one-job} we know that in $\widehat x$ there is at
  most one job (say $j$, if any) run on machine $e$. Then the
  contribution of this machine to the value of $g(\widehat \lambda)$ is
    \begin{gather}
    \left(\ell_{je} \widehat{x}_{je}\right)^\alpha - (\widehat \lambda_j - {c_{je}})
    \widehat{x}_{je}\label{3eq:9}
  \end{gather}
 Since $\widehat{x}$ is a minimizer for $g(\widehat{\lambda})$, we know that the partial derivative of the above term evaluates to zero. This gives
${\alpha  \ell_{je}} \cdot \left( \ell_{je} \widehat{x}_{je}\right)^{\alpha - 1}  -  \left(\widehat \lambda_j - {c_{je}}\right)
  = 0$,
  or equivalently,
$ {\widehat{x}}_{je} =\frac{1}{\ell_{je}} \left( \frac{\widehat \lambda_j -
      {c_{je}}}{\alpha \ell_{je}} \right)^{1/(\alpha - 1)}
 $.
  Substituting into this value of $\widehat{x}_{je}$ into equation~(\ref{3eq:9}), the contribution of machine $e$ towards the dual
 $g(\widehat \lambda)$ is
  \[  \left( \frac{\widehat \lambda_j -
      {c_{je}}}{\alpha \ell_{je}} \right)^{\alpha/(\alpha - 1)}  -
      \frac{( \widehat \lambda_j -
      {c_{je}})}{\ell_{je}}
      \left( \frac{\widehat \lambda_j -
      {c_{je}}}{\alpha \ell_{je}} \right)^{1/(\alpha - 1)} \; = \; (1-\alpha) \left( \frac{\widehat \lambda_j -
      {c_{je}}}{\alpha \ell_{je}} \right)^{\alpha/(\alpha - 1)}
  \]
  Hence, for each machine $e$, we want to choose that the job $j$ that minimizes this
  expression, which is also the job $j$ that maximizes the expression $ (\widehat \lambda_j -
    {c_{je}})/\ell_{je}$ since $\alpha > 1$. This is precisely the job $\varphi(e)$ and the proof is hence complete.  \end{proof}

\begin{lemma}
  \label{3lem:last-phi}
  For all machines $e$, job $\last(e)$ is feasible choice for $\varphi(e)$.
\end{lemma}

\begin{proof}
  The line of reasoning is the following:
  $$ \varphi(e) = \arg\max_{j} \frac{\big(\widehat \lambda_j - {c_{je}} \big)}{\ell_{je}}
  = \arg\max_{j} \bigg( \delta \cdot {\alpha} \cdot \bigg(\sum_{i\leq j}
  \ell_{je} x_{ie}\bigg)^{\alpha-1} \bigg) = \arg\max_{j} \bigg(
  \bigg(\sum_{i\leq j} \ell_{ie} x_{ie}\bigg)^{\alpha-1} \bigg) =
  \last(e) \, .$$ The first equality is the definition of
  $\varphi(e)$. For the second equality, observe that for any job $k$,
  \[ \widehat \lambda_{k} \leq \delta \cdot \alpha \cdot \ell_{ek}
  (\sum_{i \leq k} \ell_{ie} x_{ie})^{\alpha - 1} + c_{ke} \implies
  \frac{\widehat \lambda_k - c_{ke}}{\ell_{ke}} \leq \delta \,\alpha\,
  (\sum_{i \leq k} \ell_{ie} x_{ie})^{\alpha - 1}\,.\] The expression on
  the right is monotone increasing in $\sum_{i \leq k} \ell_{ie}
  x_{ie}$, the load due to jobs up to (and including $k$). Moreover, it
  is maximized by the last job to assign fractionally to $e$ (since the
  inequality is strict for all other jobs). Since this last job is
  $\last(e)$, the last equality follows.
%   follows from the definition of $\widehat \lambda_j$ (and in particular
%   the fact that if $x_{je} = 0$ then $((\widehat\lambda_j -
%   c_{je})/(\delta \alpha \ell_{je}) )^{1/(\alpha -1)}$ is less than
%   $e$'s load at the moment we finished assigning $j$), and the last
%   equality follows from the definition of $\last(e)$. 
\end{proof}

\begin{theorem}
\label{thm:finalload}
The online greedy algorithm is $\alpha^\alpha$-competitive.
\end{theorem}

\begin{proof}
By weak duality it is sufficient to show that $g(\widehat{\lambda})\ge \textrm{ON}/\alpha^\alpha$.
Applying \lref[Lemma]{3lem:phi-job} to the expression for $g(\widehat{\lambda})$ (equation~\eqref{3eq:g}) and substituting the contribution of each machine towards the dual, we get that
  \begin{equation}
  g(\widehat \lambda)
 =   \bigg( \sum_j \widehat \lambda_j +  \sum_e   (1-\alpha) \bigg( \frac{\widehat \lambda_{\last(e)} -
      {c_{\last(e)e}}}{\alpha \ell_{\last(e)e}} \bigg)^{\alpha/(\alpha - 1)}
 \bigg) \label{3eq:z0}
\end{equation}
Now we consider only the first term $\sum_j \widehat \lambda_j$ and evaluate it.
\begin{align}
\sum_j \lambda_j
& = \sum_{j,e} \widehat \lambda_j \widetilde{x}_{je} \\
& = \sum_e \widetilde{x}_{je} \bigg( \sum_j \delta \cdot {\alpha \cdot \ell_{je}} \bigg(\sum_{i\leq j}
  \ell_{ie} \widetilde{x}_{ie}\bigg)^{\alpha-1} +   c_{je}  \bigg) \\
& =  (\delta \cdot \alpha) \sum_e \sum_j  \ell_{je} \widetilde{x}_{je}  \bigg(\sum_{i\leq j}
  \ell_{ie} \widetilde{x}_{ie}\bigg)^{\alpha-1} +  \sum_{j,e} \widetilde{x}_{je} c_{je} \\
& \geq  {\delta} \sum_e \bigg( \sum_j  \ell_{je} \widetilde{x}_{je} \bigg)^\alpha +  \sum_{j,e} \widetilde{x}_{je} c_{je}
\end{align}
Now consider the second term of equation~\eqref{3eq:z0}. Note that if we substitute the value of $\widehat \lambda_{\last(e)}$, it evaluates to
${(1-\alpha)} \delta^{\alpha/(\alpha-1)} \sum_{e} \big( \sum_j \ell_{je} \widetilde{x}_{je} \big)^{\alpha}$
Putting the above two estimates together, we get
  \begin{align}
  g(\widehat \lambda)
& \geq {\delta} \sum_e \bigg( \sum_j  \ell_{je} \widetilde{x}_{je} \bigg)^\alpha +  \sum_{j,e} \widetilde{x}_{je} c_{je}
 + (1-\alpha) \delta^{\alpha/(\alpha-1)} \sum_{e} \bigg( \sum_j \ell_{je} \widetilde{x}_{je} \bigg)^{\alpha} \\
& = \bigg(\delta + (1 - \alpha) \delta^{\alpha/(\alpha-1)} \bigg)  \sum_{e} \bigg( \sum_j \widetilde{x}_{je} \ell_{je} \bigg)^{\alpha} +  \sum_{j,e} \widetilde{x}_{je} c_{je} \label{eq:ratio}\\
& \geq {\textrm{ON}}/{\alpha^\alpha}
\end{align}
The final inequality is due to the choice of $\delta = 1/\alpha^{\alpha-1}$ which maximizes $\left(\delta + (1 - \alpha) \delta^{\alpha/(\alpha-1)}\right)$.
\end{proof}

As observed, e.g., in~\cite{AAGKKV95}, this $O(\alpha)^\alpha$ result is
the best possible, even for the (fractional) \onGAP problem without any
assignment costs. In \lref[Section]{sec:rounding}, we show how to obtain
an $O(\alpha)^\alpha$-competitive algorithm for \emph{integer}
assignments by a very similar greedy algorithm, and dual-fitting, albeit
with a more careful analysis.

\section{Application to Speed Scaling}
\label{sec:speed-scaling}

We now discuss the application of our results for \onGAP to some
well-studied speed scaling problems.  In these problems a collection of jobs arrive over
time.  The $j^{th}$ job arrives at time $r_j$, and has size/work $p_j$.
These jobs must be scheduled on a speed scalable processor that can run
at any non-negative speed.  There is a convex function $P(s)=s^\alpha$
specifying the dynamic power used by the processor as a function of
speed $s$.  The value of $\alpha$ is typically around $3$ for CMOS based
processors. Commonly, one considers objectives $\cal G$ of the form
${\cal S} + \beta {\cal E}$, where $\cal S$ is a scheduling objective,
and $\cal E$ is the energy used by the system. Moreover, the scheduling
objective ${\cal S}$ is a \emph{fractional} sum objective of the form $
\sum_j \sum_t \frac{x_{jt}}{p_j} C_{jt}$, where $C_{jt}$ is the cost of
completing job a unit of work of job $j$ at time $t$, and $x_{jt}$ is the amount of work
completed at time $t$, or the corresponding \emph{integer} sum objective
$\sum_j \sum_t y_{jt} \, C_{jt}$, where $y_{jt}$ indicates whether or
not job $j$ was completed at time $t$. Fractional scheduling objectives
are interesting in their own right (for example, in situations where the
client gains some benefit from the early partial completion of a job),
and are often used in an intermediate step in the analysis of algorithms
for integer scheduling objectives.  % The intuitive rationale for such a
% sum objective can be understood as follows: Assume that the possibility
% exists to invest $E$ units of energy to decrease the scheduling
% objective by $S$ units.  Then an optimal scheduler for this objective
% would make such an investment if and only if $\beta S \ge E$.  So in
% some sense $\beta$ specifies the relative value of energy in comparison
% to the value of reducing the scheduling objective.

Normally one thinks of the online scheduling algorithm as having two
components: a \emph{job selection policy} to determine the job to run,
and a \emph{speed scaling policy} to determine the processor speed.
However, one gets a different view when one thinks of the online
scheduler as solving online the following mathematical programming
formulation of the problem (which is an instance of the \onGAP problem):
\begin{align*}
\min\quad& \sum_t \bigg(\sum_{j}  x_{jt}\bigg)^\alpha + \sum_{j} \sum_t C_{jt}\, x_{jt}  \\
\textrm{subject to}\quad&\sum_{t} x_{jt} \geq 1 \qquad j = 1, \ldots, n
\end{align*}
Here the variables $x_{jt}$ specify how much work from job $j$ is run at
time $t$. Because we are initially concerned with fractional scheduling
objectives, we can assume without loss of generality that all jobs have
unit length. The arrival of a job $j$ corresponds to the arrival of a
constraint specifying that job $j$ must be completed. Greedily raising
the primal variables corresponds to committing to complete the work of
job $j$ in the cheapest possible way, given the previous commitments.

\paragraph{Two Special Cases.}
A well-studied speed-scaling problem in this class is when the
scheduling objective is energy minimization subject to deadline
feasibility~\cite{YDS,BKP,BansalBCP11,BansalCPK09}: for each job $j$
there is a deadline $d_j$, and $c_{jt}=0$ for $t\in[r_j, d_j]$ and is
infinite otherwise. Our algorithm for \onGAP is essentially equivalent
to the algorithm \emph{Optimal Available (OA)}, introduced in \cite{YDS}
and shown to be $\alpha^\alpha$-competitive in
\cite{BKP}---specifically, the speeds set by both OA and our algorithm
are the same at all times, but the jobs that are run may be different,
since OA uses Earliest Deadline First for scheduling. Our analysis of
the online greedy algorithm for \onGAP is an alternate, and simpler,
analysis of OA than the potential function analysis in \cite{BKP}. In
this instance, our duality based analysis is tight, as OA is no better
than $\alpha^\alpha$-competitive~\cite{YDS,BKP}.

Another well-studied special case is when the objective is total
flow~\cite{AF,BPS,BansalCP09,LLTW08,AndrewLW10,ChanELLMP09,ChanLL10}.
That is, $c_{jt}$ equals $(t-r_j)$ for $t \ge r_j$ and infinite
otherwise. All prior algorithms for this objective assume some variation
of the \emph{balancing} speed scaling algorithm that sets the power
equal to the (fractional) number/weight of unfinished jobs.
Unlike the earlier example above, \onGAP behaves differently than these balancing algorithms for the following reason.
When a job arrives, \onGAP only focuses on choosing assignments which
minimize the rate of increase of the objective, and this rule determines
both the scheduling policy (in fact, the entire schedule of job $j$ is
decided upon arrival) and the power usage over time. However, in the
balancing algorithms the speed profile (and hence power usage) is
entirely determined by the scheduling policy, and the scheduling
policies used are typically the ones optimal for fractional flow like
SJF.  Hence it is likely that our algorithms will actually have a
different work profile from balancing algorithms.
% In particular, we note that for an appropriate choice of $\delta$ (not the choice
% used in our analysis), \onGAP is identical
% to this balancing algorithm for single-job instances, but will
% behave differently for non-trivial instances. (See \lref[Appendix]{egs}
% for details.)

\paragraph{A Note about our Approximation Guarantees.}
A closer examination of our analysis (especially
equation~\eqref{eq:ratio}) shows that our algorithm has a
\emph{Lagrangian multiplier preserving} property: we get that our convex
cost + $\alpha^\alpha$ times the linear term is at most $\alpha^\alpha$
times the dual. This separation between the linear and non-linear terms
in the objective happens because the constraints are linear, and when we
compute the dual, the dual variables are involved in the linear terms
whereas the convex terms in the minimizer are identical to their
expressions in the primal. In order to argue about the dual minimizer,
this somehow forces us to be exact on the linear terms. This can perhaps
explain why our analysis is tight for the deadline feasibility and load
balancing problems~\cite{BKP, AAFPW, AAGKKV95} (where there are no
linear terms), but is an $O(\alpha^\alpha)$-factor worse than the
algorithm of~\cite{BansalCP09} for fractional flow+energy (which has
non-trivial linear terms).

\paragraph{Comparison to Previous Techniques.}
In most potential function-based analyses of speed scaling
problems in the literature, the potential function is defined to be the
future cost for the online algorithm to finish the jobs if the remaining
sizes of the jobs were the lags of the jobs, which is how far the online
algorithm is behind on the job~\cite{ImMP11}.  A seemingly necessarily
condition to apply this potential function technique is that there must
be is a relatively simple algebraic expression for the future cost for
the online scheduling algorithm starting from an arbitrary state.  As it
is not clear how to obtain such an algebraic expression for the most
obvious candidate algorithms for nonlinear scheduling objectives,~\footnote{
For $L_\alpha$ norms of flow on a single fixed speed processor, no potential function is required
to prove scalability of natural online algorithms~\cite{BansalP10}.
For more complicated non-work conserving machine environments, there are analyses that use a 
potential function that is a rough approximation of future costs~\cite{GuptaIKMP10,ImM10}. But despite some effort,
it has not been clear how to extend these potential functions to apply to $L_\alpha$ norms of
flow in the speed scaling setting.} this
to date has limited the application of this potential function method to
speed scaling problems with linear scheduling objectives. However, our
dual-based analysis for \onGAP yields an online greedy speed scaling algorithm that is $O_\alpha(1)$-competitive for any sum scheduling
objective.

Our algorithm for \onGAP has the advantage that, at the release time of a job, it can commit to the
client exactly the times that each portion of the job will be run. One can certainly imagine situations
when this information would be useful to the client. Also the \onGAP analysis applies to a wider class
of machine environments than does the previous potential function based analyses in the literature. For example,
our analysis of the \onGAP algorithm can handle the case that the processor is unavailable at certain
times, without any modification to the analysis. Although this generality has the disadvantage that it
gives sub-optimal bounds for some problems, such as when the scheduling objective is total flow.

By speeding up the processor by a $(1+\epsilon)$ factor, one can obtain an online speed scaling algorithm that one can show,
using known techniques~\cite{BLMP1,BPS}, has
competitive ratios  at most $\min(\alpha^\alpha (1+\epsilon)^\alpha, \frac{1}{\epsilon})$ for the corresponding integer
scheduling objective.

\section{Routing with Speed Scalable Routers}
\label{sec:routing}

In this section we consider the following online routing
problem. Routing requests in a graph/network arriving over time.  The
$j^{th}$ request consists of a source $s_j$, a sink $t_j$, and a flow
requirement $f_j$.  In the unsplittable flow version of the problem the
online algorithm must route $f_j$ units of flow along a single $(s_j,
t_j)$-path. In the splittable flow version of this problem, the online
algorithm may partition the $f_j$ units of flow among a collection of $(s_j,
t_j)$-paths. In either case, we assume speed scalable network elements
(routers, or links, or both) that use power $\ell^\alpha$ when they have
load $\ell$, where the load is the sum of the flows through the
element. We consider the objective of minimizing the aggregate power.
We show that an intuitive online greedy algorithm is
$\alpha^\alpha$-competitive using the dual function of a mathematical
programming formulation as a lower bound to optimal.

\textbf{High Level Idea:} The proof will follow the same general
approach as for \onGAP: we define dual variables $\widehat\lambda_j$ for
the demand pairs, but now the minimization problem (which is over flow
paths, and not just job assignments) is not so straight-forward: the
different edges on a path $p$ might want to set $f(p)$ to different
values. So we do something seemingly bad: we \emph{relax} the dual to
decouple the variables, and allow each (edge, path) pair to choose its
own ``flow'' value $f(p,e)$. And which of these should we use as our
surrogate for $f(p)$? We use a convex combination $\sum_{e \in p} h_e\,
f(p,e)$---where the multipliers $h(e)$ are chosen \emph{based on the
  primal loads}(!), hence capturing which edges are important and which
are not.

\subsection{The Algorithm and Analysis}

We first consider the splittable flow version of the problem. Therefore, we can
assume without loss of generality that all flow requirements are unit,
and all sources and sinks are distinct (so we can associate a unique
request $j(p)$ with each path $p$). This will also allow us to order
paths by the order in which flow was sent along the paths. We now  model
the problem by the following primal optimization formulation:
\begin{align*}
\min\quad& \sum_e \bigg(\sum_j \sum_{p  \ni e: p \in P_j}  f(p)\bigg)^\alpha  \\
\textrm{subject to}\quad&\sum_{p \in P_j} f(p) \geq 1 \qquad j = 1, \ldots, n
\end{align*}
where $P_j$ is the set of all $(s_j, t_j)$ paths,
and $f(p)$ is a non-negative real variable denoting the amount of flow routed on the path $p$.
In this case, the dual function is:
\begin{align*}
g(\lambda) =   \min_{f(p)} \bigg( \sum_j \lambda_j +  \sum_e \bigg(\sum_j \sum_{p  \ni e: p \in P_j}  f(p)\bigg)^\alpha
-\sum_{j, p \in P_j} \lambda_j f(p) \bigg)
\end{align*}
One can think of the dual function as a routing problem with the same instance, but without the constraints
that at least a unit of flow must be routed for each request.
In the objective, in addition to energy costs, a fixed cost of $\lambda_j$ is payed for each request $j$, and
a payment of $\lambda_j$ is received for each unit of flow routed from $s_j$ to $t_j$.

\medskip
\noindent
{\bf Description of the Online Greedy Algorithm:} To route flow for request $j$, flow is continuously routed along the paths that will increase costs the
least until enough flow is routed to satisfy the request. That is, flow is routed along all $(s_j, t_j)$ paths $p$ that
minimize $\sum_{e \in p}  \alpha \cdot \left( \sum_{q \le p: q \ni e} f(q) \right)^{\alpha -1}$.
For analysis purposes, after the flow for request $j$ is routed, we define
$$\widehat \lambda_j = \alpha \delta \bigg(\sum_{e \in p} \sum_{q \le p: q \ni e }  f(q)\bigg)^{\alpha-1} $$
where $p$ is any path along which flow for request $j$ was routed, and
$\delta$ is a constant (later set to $\frac{1}{\alpha^{\alpha -1}}$).
\medskip

\textbf{The Analysis:}
Unfortunately, unlike the previous section for load balancing, it is not so clear how to compute the dual function $g(\widehat \lambda)$  or its minimizer since the variables cannot be nicely decoupled as we did there (per machine). In order to circumvent this difficulty,
we consider the following \emph{relaxed} function $\widehat g(\widehat \lambda, h)$, which does not enforce the constraint that flow must
be routed along paths. This enables us to decouple variables and then argue about the objective value.
Indeed, let $\widetilde f(p)$ be the final flow on path $p$ for the routing produced by the online algorithm.
Let $h(e) = \alpha \sum_{p \ni e} \widetilde f(p)^{\alpha -1}$ be the incremental cost of routing additional flow along edge $e$,
and $h(p) = \sum_{e \in p} h(e)$ be the incremental cost of routing additional flow along path $p$. We then define:
\begin{align*}
\widehat g(\widehat \lambda, h) =   \min_{f(p,e)} \bigg(  \sum_j \widehat \lambda_j +\sum_e \bigg(\sum_j \sum_{p \ni e: p \in P_j}  f(p, e)\bigg)^\alpha
-\sum_{j} \widehat \lambda_j \sum_{p \in P_j} \sum_{e \in p} \frac{h(e)}{h(p)} f(p,e) \bigg)
\end{align*}
Conceptually, $f(p,e)$ can be viewed as the load placed on edge $e$ by request $j(p)$.
In $\widehat g(\widehat \lambda, h)$, the scheduler has the option of increasing the load on individual edges $e \in p \in P_j$, but the income
from edge $e$ will be a factor of $\frac{h(e)}{h(p)}$ less than the income achieved in $g(\widehat \lambda)$.
In  \lref[Lemma]{claim:relaxg} we prove that $\widehat g(\widehat \lambda, h)$ is
a lower bound for $g(\widehat \lambda)$.
We then proceed as in the analysis of \onGAP.
 \lref[Lemma]{lemma:computedual} shows how the minimizer and value of $\widehat g(\widehat \lambda, h)$
can be computed, and \lref[Lemma]{claim:jh} shows how to bound some of the dual variables
in terms of the final online primal solution.

\begin{lemma}
\label{claim:relaxg}
For the above setting of $h(\cdot)$, $\widehat g(\widehat \lambda, h) \le g(\widehat \lambda)$.
\end{lemma}
\begin{proof}
We show that there is a feasible value of $\widehat g(\widehat \lambda, h)$ that is less than $g(\widehat \lambda)$.
Let the value of  $f(p, e)$ in $\widehat g(\widehat \lambda, h)$ be the same as the value of $f(p)$ in $g(\widehat \lambda)$.
Plugging these values for $f(p, e)$ into the expression for $\widehat g(\widehat \lambda, h)$, and simplifying, we get:
\begin{align*}
\widehat g(\widehat \lambda, h) &\le \sum_j \widehat \lambda_j +
\sum_e \bigg(\sum_j \sum_{p  \ni e: p \in P_j}  f(p)\bigg)^\alpha
-\sum_{j} \widehat \lambda_j \sum_{p \in P_j} f(p) \sum_{e \in p}  \frac{h(e)}{h(p)} \\
&= \sum_j \widehat \lambda_j +\sum_e \bigg(\sum_j \sum_{p  \ni e: p \in P_j}  f(p)\bigg)^\alpha
-\sum_{j} \widehat \lambda_j \sum_{p \in P_j} f(p) \\
&= g(\widehat \lambda)
\end{align*}
The first equality holds by the definitions of $h(e)$ and $h(p)$, and the second equality holds by the optimality of $f(p)$.
\end{proof}

\begin{lemma}
 \label{lemma:computedual}
 There is a minimizer $\widehat{f}$ of $\widehat g(\widehat \lambda, h)$ in which for each edge $e$, there is a
 single path $p(e)$ such that $\widehat f(p,e)$ is positive, and
 $\widehat f(p(e), e) = \left(\frac{\widehat \lambda_{j(p(e))}h(e)}{\alpha \cdot h(p(e))} \right)^{1/(\alpha-1)}$.
 \end{lemma}
 \begin{proof} The argument that  $\widehat g(\widehat \lambda, h)$ has a minimizer where each
 edge only has flow from one request follows the same reasoning as in the proof of Lemma
  \ref{3lem:one-job}.
Once we know only one request sends flow on any edge we can use calculus to identify the minimizer and the value which achieves it. Indeed, we get
that
$p(e) = \arg \max_{p \ni e} \frac{ \widehat \lambda_{j(p)} h(e)}{h(p)}$.
and the value of $\widehat f(p(e), e)$ is set so that
the incremental energy cost would just offset the incremental income from routing the flow, that is
$\alpha \, \widehat f(p(e), e)^{\alpha -1}  = \widehat \lambda_{j(p(e))} \frac{h(e)}{h(p(e))}$.
Solving for  $\widehat f(p(e), e)$, the result follows.
\end{proof}

\begin{lemma}
\label{claim:jh}
$\widehat \lambda_{j(p(e))} \le \delta \cdot h(p(e))$
\end{lemma}
\begin{proof}
$\widehat \lambda_{j(p(e))}$ is $\delta$ times the rate at which the energy cost was increasing for
the online algorithm when it routed the last bit of flow for request $j(p(e))$.
$h(p(e))$ is the rate at which the energy cost would increase for the online algorithm if additional flow was pushed along
$p(e)$ after the last request was satisfied.  If $p(e)$ was a path on which the online algorithm routed flow,
then the result follows from the fact the online algorithm never decreases the flow on any edge.
If $p(e)$ was not a path on which the online algorithm routed flow, then the result follows from the fact
that at the time that the online algorithm was routing flow for request $j(p(e))$, $p(e)$ was more costly than
the selected paths, and the cost can't decrease subsequently due to monotonicity of the flows in the online solution.
\end{proof}

\begin{theorem}
\label{thm:route}
The online greedy algorithm is $\alpha^\alpha$ competitive.
\end{theorem}

\begin{proof}
We will show that $\widehat g(\widehat \lambda, h)$ is at least the online cost ON divided by $\alpha^\alpha$,
which is sufficient since $\widehat g(\widehat \lambda, h)$ is a lower bound to $g(\widehat \lambda)$ by
\lref[Lemma]{claim:relaxg}, and since $g(\widehat \lambda)$ is a lower bound to optimal.
\begin{align}
\widehat g(\widehat \lambda, h) &=   \min_{f(p,e)} \bigg(  \sum_j \widehat \lambda_j +\sum_e \bigg(\sum_j \sum_{p \ni e: p \in P_j}  f(p, e)\bigg)^\alpha
-\sum_{j} \widehat \lambda_j \sum_{p \in P_j} \sum_{e \in p} \frac{h(e)}{h(p)} f(p,e) \bigg) \label{eq:q0}\\
&= \sum_j \widehat \lambda_j  - (\alpha - 1)
\sum_e \bigg(\frac{\widehat \lambda_{j(p(e))}h(e)}{\alpha \cdot h(p(e))} \bigg)^{\alpha/(\alpha-1)} \label{eq:q1}\\
&\ge \sum_j \widehat \lambda_j  - (\alpha - 1)
\sum_e \bigg(\frac{\delta \cdot h(e)}{\alpha } \bigg)^{\alpha/(\alpha-1)} \label{eq:q2}\\
&= \sum_j \widehat \lambda_j  - (\alpha - 1) \delta^{\alpha/(\alpha-1)}
\sum_e \bigg(\sum_{p \ni e} \widetilde f(p) \bigg)^\alpha \label{eq:q3}\\
&= \sum_j \widehat \lambda_j  \sum_{p \in P_j} \widetilde f(p) - (\alpha - 1) \delta^{\alpha/(\alpha-1)}
\sum_e \bigg(\sum_{p \ni e} \widetilde f(p) \bigg)^\alpha \label{eq:q4}\\\
&=   \delta \alpha \sum_j  \sum_{p \in P_j} \widetilde f(p) \bigg(\sum_{e \in p} \sum_{q \le p: q \ni e} \widetilde f(q)\bigg)^{\alpha-1}- (\alpha - 1) \delta^{\alpha/(\alpha-1)}
\sum_e \bigg(\sum_{p \ni e} \widetilde f(p) \bigg)^\alpha \label{eq:q5}\\
&\ge \delta \sum_e \bigg(\sum_{p \ni e} \widetilde f(p) \bigg)^\alpha - (\alpha - 1) \delta^{\alpha/(\alpha-1)}
\sum_e \bigg(\sum_{p \ni e} \widetilde f(p) \bigg)^\alpha\label{eq:q7}\\
&= \frac{1}{\alpha^\alpha} \sum_e \bigg(\sum_{p \ni e} \widetilde f(p) \bigg)^\alpha\label{eq:q8}\\
&\ge \textrm{ON}/\alpha^\alpha
\end{align}
The equality in line (\ref{eq:q0}) is the definition of $\widehat g(\widehat \lambda, h)$.
The equality in line (\ref{eq:q1}) follows from Lemma  \ref{lemma:computedual}.
The inequality in line (\ref{eq:q2}) follows from \lref[Lemma]{claim:jh}.
The equality in line (\ref{eq:q3}) follows from the definition of $h(e)$.
The equality in line (\ref{eq:q4}) follows from the feasibility of $\widetilde f$.
The equality in line (\ref{eq:q5}) follows from the definition of $\widehat \lambda$.
The equality in line (\ref{eq:q7}) follows from the definition of $\delta$.
\end{proof}

%We now show that an $O(1)$-competitive algorithm for the unsplittable flow
%problem can be obtained by applying the most obvious online randomized rounding
%procedure, which routes the flow $f_j$ along the path $p \in P_j$ with probability
%$\widetilde f(p)$. Denote the resulting integer variables by $F(p)$.
%Let the resulting load on element $e$ be denoted by
%$L_e = \sum_{p \ni e} F_{p}$, which is
%%the sum of non-negative and independent random variables. Hence, by
%Rosenthal's inequality~\cite{Rosen, JSZ},
%\[ E[L_e^\alpha] \leq k_\alpha \max\bigg( \sum_{p \ni e} f_{j(p)}^\alpha \widetilde f(p),
%\bigg(\sum_{p \ni e}
%f_{j(p)} \widetilde f(p)\bigg)^\alpha \bigg) \]
%where $k_\alpha$ is a constant depending only on $\alpha$.

While the above algorithm only gives a splittable routing, i.e., a fractional routing, we note that the ideas of the next section, \lref[Section]{sec:rounding}, can be used to obtain an $O_\alpha(1)$-competitive algorithm for \emph{integer flow}
as well by using a slightly modified primal program (we have to handle non-uniform demands, and also strengthen the basic convex program to prevent some trivial integrality gaps. The next section described how we can handle these issues for the load balancing problem. 

\section{Online Load Balancing: Integral Assignments}
\label{sec:rounding}

For simplicity, let us consider online integer load balancing
\emph{without} assignment costs; it is easy to see the extension to the
other problems that we consider.  In this problem each job has the
values $\ell_{je}$, and the goal is to integrally assign it to a single
machine so as to minimize the sum $\sum_e (\sum_j X_{je}
\ell_{je})^\alpha$ where $X_{je}$ is the indicator variable for whether
job $j$ is assigned to machine~$e$.  The most natural reduction to our
general model \onGAP is to set all $c_{je}$'s to $0$.  However, the
convex relaxation for this setting has a large integrality gap with
respect to integral solutions. For example, consider the case of just a
single job which splits into $m$ equal parts (where $m$ is the number of
machines)---the integer primal optimal pays a factor of $m^{\alpha-1}$
times the fractional primal optimal. To handle this case, we add a fixed
assignment cost of $c_{je} = \ell_{je}^\alpha$ for assigning job $j$ to
machine~$e$. It is easy to see that the cost of an optimal integral
solution at most doubles in this relaxation. This is the convex program
that we use for the rest of this section.

\subsection{Approach I: Integer Assignment}
\label{sec:integer1}

In this section, we show that an algorithm which gives an
$O(\alpha)^\alpha$-competitive ratio.  Consider the following greedy
algorithm: when job $j$ arrives, it picks the machine $e$ that minimizes
\[ \delta\cdot \alpha \cdot \ell_{je} ( \sum_{i < j} \ell_{ie}
x_{ie})^{\alpha - 1} + \ell_{je}^\alpha,\] and set $x_{je} =
1$. Moreover, set $\widehat\lambda_j$ for job $j$ to be precisely the
quantity above. Let $\widetilde x_{ie}$ be the final settings of the
primal variables.

For the analysis, we again need to set the $\widehat x_j$'s. For machine
$e$, again consider $\varphi(e)$ and $\last(e)$ as defined in the
previous proofs. We can no longer claim that $\last(e)$ is a feasible
setting for $\varphi(e)$. However, we can claim that the load on machine
$e$ that is seen by $\varphi(e)$ (\emph{and indeed, by any job $k$}) is
at most the load seen by $\last(e)$ when it arrived, \emph{plus} the length
$\ell_{\last(e)e}$. Hence
\begin{gather}
  \widehat \lambda_{\varphi(e)} \leq \delta \, \alpha\,
  \ell_{\varphi(e)e} ( \sum_{i < \last(e)} \ell_{ie} \widetilde{x}_{ie} +
  \ell_{\last(e)e})^{\alpha - 1} + \ell_{\varphi(e)e}^\alpha \implies
  \frac{\widehat \lambda_{\varphi(e)} -
    \ell_{\varphi(e)e}^\alpha}{\alpha\ell_{\varphi(e)e}} \leq \delta \,
  ( \sum_{i \leq \last(e)} \ell_{ie} \widetilde{x}_{ie})^{\alpha - 1}.\label{eq:1}
\end{gather}
But since $\last(e)$ is the last job on machine $e$, this last
expression is exactly $\delta \, ( \sum_{i} \ell_{ie} \widetilde{x}_{ie})^{\alpha -
  1}$.  Now recall the dual from~(\ref{3eq:z0}). Observing that that
$\alpha > 1$, we can use the calculations we just did to get
\begin{align*}
  g(\widehat \lambda) \geq& \sum_{j} \widehat \lambda_j + (1 - \alpha)
  \delta^{\alpha/(\alpha - 1)}\,
  \sum_e \bigg(
  \sum_{i} \ell_{ie} \widetilde{x}_{ie}
  \bigg)^{\alpha} \\
  =&   \alpha \delta\,
  \sum_{e,j} \ell_{je} \widetilde{x}_{je} \bigg( \sum_{i: i < j } \ell_{ie} \widetilde{x}_{ie}
  \bigg)^{\alpha-1} + \sum_{j,e} \ell_{je}^\alpha  \widetilde{x}_{je}
  + (1-\alpha) \delta^{\alpha/(\alpha - 1)}\,
  \bigg( \sum_{j, e} \ell_{je} \widetilde{x}_{je} \bigg)^{\alpha} \\
  =&  \sum_{e} \left( \alpha \delta\,
    \sum_{j \in S_e} \ell_{je}  \bigg( \sum_{i \in S_e: i < j } \ell_{ie} 
    \bigg)^{\alpha-1} + \sum_{j \in S_e} \ell_{je}^\alpha 
    + (1-\alpha) \delta^{\alpha/(\alpha - 1)}\,
    \bigg( \sum_{j \in S_e} \ell_{je} \bigg)^{\alpha} \right) \\
  \geq& \frac{1}{e(e(\alpha+1))^\alpha} \times
  \sum_e \bigg( \sum_{j \in S_e} \ell_{je} \bigg)^{\alpha}
\end{align*}
where the last inequality is obtained by applying
\lref[Lemma]{lem:calculation} to the expression for each machine
$e$. This implies the $O(\alpha)^\alpha$ competitive ratio for the
integral assignment algorithm.

\begin{lemma}
  \label{lem:calculation}
  Given non-negative numbers $a_0, a_1, a_2, \ldots, a_T$ and $\delta =
  (e(\alpha+1))^{\alpha-1}$, we get 
  \begin{gather}
    \alpha \delta\, \sum_{j \in [T]} a_j \bigg( \sum_{i < j } a_i
    \bigg)^{\alpha-1} + \sum_{j \in [T]} a_j^\alpha  +  (1-\alpha)
    \delta^{\alpha/(\alpha - 1)}\, \bigg( \sum_{j \in [T]} a_j
    \bigg)^{\alpha} ~~ \geq ~~ \frac{1}{e(e(\alpha+1))^\alpha} \times
    \bigg( \sum_{j \in [T]} a_j \bigg)^{\alpha} . \label{eq:4} 
  \end{gather}
\end{lemma}

\begin{proof}
  First, consider the case when $\alpha \geq 2$: in this case, we bound
  the LHS of~(\ref{eq:4}) when the sequence of numbers is
  non-decreasing, and then we show the non-decreasing sequence makes
  this LHS the smallest. We then consider the (easier) case of $\alpha
  \in [1,2]$.

  Suppose $a_0 \leq a_1 \leq \cdots \leq a_T$, then $\sum_{j = 0}^T a_j
  ( \sum_{i < j } a_i )^{\alpha-1} \geq \sum_{j =0}^{T-1} a_j ( \sum_{i
    \leq j } a_i )^{\alpha-1}$, and it suffices to (lower) bound the following term
  \begin{gather}
     \alpha \delta\, \sum_{j = 0}^{T-1} a_j \bigg( \sum_{i: i \leq j } a_i
    \bigg)^{\alpha-1} \quad +\quad \sum_{j = 1}^T a_j^\alpha \quad - \quad
    (\alpha-1) \delta^{\alpha/(\alpha - 1)}\,
    \bigg( \sum_{j=1}^T a_j \bigg)^{\alpha} \label{eq:3}
  \end{gather}

  There are two cases, depending on the last term: whether $a_T \leq
  \frac1\alpha \sum_{j = 0}^{T-1} a_j$, or not.
  \begin{itemize}
  \item If $a_T \leq \frac1\alpha \sum_{j= 0}^{T-1} a_j$, we get
    $\sum_{j = 0}^{T} a_j \leq (1+1/\alpha)\sum_{j=0}^{T-1} a_j$.  Now,
    consider the first term in (\ref{eq:3}):
    \[ \alpha \delta\, \sum_{j = 0}^{T-1} a_j \bigg( \sum_{i: i \leq j }
    a_i \bigg)^{\alpha-1} \geq \delta \bigg( \sum_{j = 0}^{T-1} a_j
    \bigg)^{\alpha} \geq \frac{\delta}{(1+1/\alpha)^\alpha} \bigg(
    \sum_{j = 0}^{T} a_j \bigg)^{\alpha} \geq \frac{\delta}{e} \bigg(
    \sum_{j = 0}^{T} a_j \bigg)^{\alpha} . \] (The last inequality used
    the fact that $(1+1/\alpha)^\alpha$ approaches $e$ from below.)
    Finally, plugging this back into~(\ref{eq:3}), ignoring the second
    sum, and using $\delta = (e(\alpha+1))^{1 - \alpha}$, we can lower
    bound the expression of~(\ref{eq:3}) by $(\sum_{j = 0}^{T} a_j)^{\alpha}$ times
    \[ \frac{\delta}{e} - (\alpha - 1)\delta^{\alpha/(\alpha -1)} =
    \frac{1}{e(e(\alpha+1))^{\alpha - 1}} - \frac{\alpha -
      1}{(e(\alpha+1))^{\alpha}} = \frac{(\alpha + 1) - (\alpha -
      1)}{(e\alpha)^\alpha(1+1/\alpha)^{\alpha}} \geq
    \frac{2}{e(e\alpha)^\alpha} \]

  \item In case $a_T \geq \frac1\alpha \sum_{j =0}^{T-1} a_j$, we get
    $\sum_j a_j = \sum_{j < T} a_j + a_T \leq (1+\alpha) a_T$. Now using
    just the single term $a_T^\alpha$ from the first two summations
    in~(\ref{eq:3}), we can lower bound it  by
    \[ 
    a_T^\alpha - (\alpha-1) \delta^{\alpha/(\alpha - 1)}
    (1+\alpha)^\alpha a_T^\alpha = a_T^\alpha \left( 1 -
      \frac{(\alpha-1)(1+\alpha)^\alpha}{(e(\alpha+1))^\alpha} \right) =
    a_T^\alpha \left(1 - \frac{\alpha-1}{e^\alpha}\right) .\] This is at
    least $a_T^\alpha/2 \geq \frac{1}{2(1+\alpha)^\alpha}\, (\sum_j
    a_j)^\alpha \geq \frac{1}{2e\alpha^\alpha}$.
  \end{itemize}
  So in either case the inequality of the statement of Lemma~\ref{lem:calculation} is satisfied.

  Now to show that the non-decreasing sequence makes the LHS smallest
  for $\alpha \geq 2$. Only the first summation depends on the order, so
  focus on $\sum_{j \in [T]} a_j ( \sum_{i < j } a_i
  )^{\alpha-1}$. Suppose $a_k > a_{k+1}$, then let $a_k = l$, $a_{k+1} =
  s$; moreover, we can scale the numbers so that $\sum_{i < k} a_i =
  1$. Now swapping $a_k$ and $a_{k+1}$ causes a decrease of
  \[ (a_k - a_{k+1}) \cdot 1^{\alpha - 1} + a_{k+1} (1 + a_k)^{\alpha -
    1} - a_k (1+a_{k+1})^{\alpha - 1} = (l-s) + s(1+l)^{\alpha - 1} -
  l(1+s)^{\alpha - 1}.  \]
  And for $\alpha \geq 2$ this quantity is non-negative. 
  %\agnote{Put in  a proof here.}

  Finally, for the case $\alpha \in [1,2)$. Note that $\alpha \delta
  \leq 1$ for our choice of $\delta$, so  the LHS
  of~(\ref{eq:4}) is at least
  \begin{align*}
    & ~ \alpha \delta\, \sum_{j \in [T]} a_j \left( \bigg( \sum_{i < j } a_i
    \bigg)^{\alpha-1} +  a_j^{\alpha-1} \right) ~ + ~ (1-\alpha)
    \delta^{\alpha/(\alpha - 1)}\, \bigg( \sum_{j \in [T]} a_j
    \bigg)^{\alpha} \\
    \geq & ~\alpha \delta \, \sum_{j \in [T]} a_j  \bigg( \sum_{i \leq j } a_i
    \bigg)^{\alpha-1} ~ + ~ (1-\alpha)
    \delta^{\alpha/(\alpha - 1)}\, \bigg( \sum_{j \in [T]} a_j
    \bigg)^{\alpha}
  \end{align*}
  The second inequality used the fact that $a^\beta +
  b^\beta \geq (a+b)^\beta  $ for $\beta \in (0,1)$. Now the proof
  proceeds as usual and gives us the desired $O(e\alpha)^\alpha$ bound.
\end{proof}

\subsection{Approach II: Randomized Rounding}
\label{sec:integer2}

We now explain a different way of obtaining integer solutions: by
rounding (in an online fashion) the fractional solutions obtained for
the problems that we consider into integral solutions. While this has a
weaker result, it is a simple strategy that may be useful in some contexts. 

Suppose we have a fractional solution w.r.t the above parameters (after
including the assignment
cost). %We note that for the load balancing problem,~\cite{AAGKKV95} give a $O(p)$-competitive online
%algorithm based on defining a clever potential function for the specific problem.
While it is known that the convex programming formulation we use has an
integrality gap of $2$~\cite{AE05,KMPS05}, these proofs use correlated
rounding procedures which we currently are not able to implement online.
Instead we analyze the simple online rounding procedure that independently and randomly
rounds the fractional assignment.
% . We show that we can indeed recover
% an integer solution while losing an additional factor of $O(p)$ in the
% approximation ratio by this
% technique.
% , even the broadly applicable technique of independent randomized rounding works.
%In this section, we show that by simply performing randomized rounding according to the $x_{je}$ values gives an integral schedule with $L_\alpha$ objective at most $O(p)$ times the online fractional solution's objective thus giving us an $O(\alpha^2)$-competitive algorithm.
Indeed, suppose we independently assign each
job $j$ to a machine $e$ with probability $\widetilde{x}_{je}$. Denote the
integer assignment induced by this random experiment by $Y_{je}$.
 Let $L_e = \sum_j \ell_{je} Y_{je}$ denote the random load on machine $e$ after the independent randomized rounding. Note that $L_e$ is
a sum of non-negative and independent random variables. We can now use the following inequality (for bounding higher moments of sums of random variables) due to Rosenthal~\cite{Rosen, JSZ} to get that
$$ \E\left[L_e^\alpha\right]^{1/\alpha} \leq K_\alpha \max\bigg( \sum_j \E\left[\ell_{je} Y_{je}\right], \bigg(\sum_j
\E\left[\ell_{je}^\alpha Y_{je}^\alpha\right]\bigg)^{1/\alpha} \bigg), $$ where $K_\alpha = O(\alpha/\log \alpha)$. However we know that $\sum_j \E\left[\ell_{je} Y_{je}\right] = \sum_j \ell_{je} \widetilde{x}_{je}$,
and that $\sum_j E[\ell_{je}^\alpha Y_{je}^\alpha] = \sum_j E[\ell_{je}^\alpha Y_{je}] =
\sum_j \ell_{je}^\alpha \widetilde{x}_{je}$. Substituting this back in and
using $(a+b)^\alpha \leq 2^{\alpha-1}(a^\alpha+b^\alpha)$, we get
\[ \E[L_e^\alpha] \leq \bigg( K_\alpha\, \max\bigg( \sum_j \ell_{je}
\widetilde{x}_{je}, \bigg(\sum_j \ell_{je}^\alpha
\widetilde{x}_{je}\bigg)^{1/\alpha}\bigg) \bigg)^\alpha \leq
2^{\alpha-1} K_\alpha^\alpha \bigg( \big( \sum_j \ell_{je}
\widetilde{x}_{je} \big)^\alpha + \sum_j c_{je} x_{je} \bigg). \]
Summing over all $e$, we infer that $\E[\sum_e L_e^\alpha]$ is at most
$(2K_\alpha)^\alpha$ times the value of the online fractional solution
objective, and hence at most $(2\alpha K_\alpha)^\alpha =
O(\alpha^2/\log \alpha)^\alpha$ times the integer optimum by
\lref[Theorem]{thm:finalload}. (Note that the results of the previous
section, and those of ~\cite{AAGKKV95,Carag08} give
$O(\alpha)^\alpha$-competitive online algorithms for the integer case.)

\section{Conclusion}

The online primal-dual dual technique (surveyed in \cite{BN-mono}) has proven to 
be a widely-systematically-applicable method to analyze online algorithms for problems expressible by linear programs.
This paper develops an analogous technique to  analyze online algorithms for problems expressible by 
\emph{nonlinear} programs. The main difference is
that in the nonlinear setting one can not disentangle the objective and the
constraints in the dual, and hence
the arguments for the dual have a
somewhat different feel to them than in the linear setting. 
We apply this technique to several natural nonlinear covering problems, most notably obtaining competitive analysis for
greedy algorithms for uniprocessor speed scaling problems with essentially arbitrary scheduling objectives
that researchers were not previously able to analyze using the prevailing potential function based analysis techniques.

Independently and concurrently with this work, Anand, Garg and Kumar \cite{AGK12} obtained results that are
in the same spirit as the results obtained here. Mostly notably, they showed
how to use nonlinear-duality to analyze
a greedy algorithm for a multiprocessor speed-scaling problem involving minimizing flow plus energy on
unrelated machines. More generally, \cite{AGK12} showed how duality based analyses could be given for
several scheduling algorithms that were analyzed in the literature using potential functions.

\bibliographystyle{alpha}
{\small \bibliography{pd}}

\end{document}